\documentclass[11pt]{article}

\usepackage[T1]{fontenc}
\usepackage{lmodern}
\usepackage{hyperref}
\usepackage[letterpaper,margin=1.00in]{geometry}
\usepackage{amsmath, amssymb, amsthm, thmtools, amsfonts}
\usepackage{bbm}
\usepackage{cite}
\usepackage{appendix}
\usepackage{graphicx}
\usepackage{color}
\usepackage{algorithm}
\usepackage[noend]{algpseudocode}
\usepackage{xspace}
\usepackage{mathtools}
\usepackage{todonotes}
\usepackage{booktabs}
\usepackage{enumerate}
\usepackage{enumitem}
\newcommand{\ra}[1]{\renewcommand{\arraystretch}{#1}}

\newtheorem{theorem}{Theorem}[section]
\newtheorem{lemma}[theorem]{Lemma}
\newtheorem{fact}[theorem]{Fact}
\newtheorem{meta-theorem}[theorem]{Meta-Theorem}

\newtheorem{corollary}[theorem]{Corollary}

\newtheorem{observation}[theorem]{Observation}
\newtheorem{definition}[theorem]{Definition}

\definecolor{darkgreen}{rgb}{0,0.5,0}
\definecolor{darkblue}{rgb}{0,0,0.5}
\usepackage{hyperref}
\hypersetup{
    unicode=false,
    colorlinks=true,
    linkcolor=darkblue,
    citecolor=darkgreen,
    filecolor=magenta,
    urlcolor=cyan
}
\usepackage[capitalize, nameinlink]{cleveref}
\crefname{theorem}{Theorem}{Theorems}
\Crefname{lemma}{Lemma}{Lemmas}
\Crefname{fact}{Fact}{Facts}
\Crefname{observation}{Observation}{Observations}
\Crefname{remark}{Remark}{Remarks}
\Crefname{invariant}{Invariant}{Invariants}
\Crefname{equation}{}{}

\newcommand{\s}{\mathit{State}}
\newcommand{\inactive}{\mathit{Inactive}}
\newcommand{\beep}{\mathit{Beep}}
\newcommand{\pulse}{\mathit{Pulse}}

\newcommand{\listen}{\mathit{Listen}}
\newcommand{\induced}{\mathit{Induced}}
\newcommand{\lock}{\mathit{Lock}}
\renewcommand{\paragraph}[1]{\vspace{0.15cm}\noindent {\bf #1}.}

\title{Time- and Space-Optimal Clock Synchronization in the Beeping Model\footnote{A version of this paper is to appear ar SPAA 2020.} \footnote{This work has been supported by the DFG Project SFB 901 (On-The-Fly Computing) and the DFG Project SCHE 1592/6-1 (PROGMATTER).
}
}

\author{
  Michael Feldmann\\
  \small Paderborn University \\
  \small michael.feldmann@upb.de\\
  \and
  Ardalan Khazraei \\
  \small Hasso Plattner Institute \\
  \small Ardalan.Khazraei@hpi.de\\
  \and
  Christian Scheideler\\
  \small Paderborn University \\
  \small scheideler@upb.de\\
}

\date{}

\begin{document}

\begin{titlepage}

\maketitle
\thispagestyle{empty}

\begin{abstract}
	We consider the clock synchronization problem in the (discrete) beeping model: Given a network of $n$ nodes with each node having a clock value $\delta(v) \in \{0,\ldots T-1\}$, the goal is to synchronize the clock values of all nodes such that they have the same value in any round.
	As is standard in clock synchronization, we assume \emph{arbitrary activations} for all nodes, i.e., the nodes start their protocol at an arbitrary round (not limited to $\{0,\ldots,T-1\}$).
	
	We give an asymptotically optimal algorithm that runs in $4D + \Bigl\lfloor \frac{D}{\lfloor T/4 \rfloor} \Bigr \rfloor \cdot (T \mod 4) = O(D)$ rounds, where $D$ is the diameter of the network.
	Once all nodes are in sync, they beep at the same round every $T$ rounds.
	The algorithm drastically improves on the $O(T D)$-bound of \cite{firefly_sync} (where $T$ is required to be at least $4n$, so the bound is no better than $O(nD)$).
	Our algorithm is very simple as nodes only have to maintain $3$ bits in addition to the $\lceil \log T \rceil$ bits needed to maintain the clock.
	
	Furthermore we investigate the complexity of \emph{self-stabilizing} solutions for the clock synchronization problem: We first show lower bounds of $\Omega(\max\{T,n\})$ rounds on the runtime and $\Omega(\log(\max\{T,n\}))$ bits of memory required for any such protocol.
	Afterwards we present a protocol that runs in $O(\max\{T,n\})$ rounds using at most $O(\log(\max\{T,n\}))$ bits at each node, which is asymptotically optimal with regards to both, runtime and memory requirements.
\end{abstract}

\end{titlepage}

\section{Introduction}
Biologically inspired algorithms try to model the behavior of certain phenomena that occur in nature.
Examples for this can be found for ants \cite{DBLP:conf/wdag/CornejoDLN14, DBLP:conf/podc/GhaffariMRL15}, cuckoos \cite{DBLP:journals/ewc/GandomiYA13}, bats  \cite{DBLP:conf/sibgrapi/NakamuraPCRPY12} and many more.
For an overview consider for example \cite{yk13}.
Especially fireflies \cite{Smith151, Mirollo:1990:SPB:95552.95570} have drawn the attention of the scientific community, as their ability to synchronize their lights can have interesting applications for programmable matter or wireless ad-hoc networks.

In this work we consider the (discrete) clock synchronization problem in the beeping model with arbitrary activations.
In the beeping model one assumes that the communication between nodes is limited to \emph{beeps} (i.e., to one single bit per round) and that a node is only allowed to broadcast a beep to all of its neighbors without knowing any information on its neighbors and, particularly, how many neighbors there are.
A node listening to a beep is not able to determine which of its neighbors, or even if multiple neighbors, generated the beep.
Arbitrary activation means that nodes do not start their protocol at the same time (which would make the problem trivial to solve), but are activated by an adversary in arbitrary rounds or if one of their neighboring nodes beeps.
Applications for this scenario can be found, for example, in (nano-) robotic systems, where multiple robots are spread in a plane.
Each robot can be seen as a weak device with only limited amount of storage capacity and a weak signal range, i.e., it is not able to communicate to all other robots but only to those that are close enough.

We present a fast algorithm whose runtime is asymptotically optimal for this problem and also investigate protocols that are additionally self-stabilizing: A self-stabilizing protocol is able to recover the system from any transient faults like message loss or blackout of nodes, which often occur in large systems.
Here we first prove a lower bound on the runtime and memory consumption for any self-stabilizing protocol and then present a self-stabilizing algorithm that is asymptotically optimal w.r.t. the aforementioned lower bound for both, time and memory complexity.

\subsection{Model and Problem Statement}
We adapt the discrete beeping model as it was introduced in \cite{DBLP:conf/wdag/CornejoK10}.
The communication network is represented by a connected undirected graph $G=(V,E)$ of $n$ nodes.
Denote the diameter of $G$ by $D$ and define $N(v) = \{w \in V\ |\ \{v,w\} \in E\}$ to be the set of neighbors for a node $v \in V$.
Time is divided into synchronous rounds and all nodes share knowledge of an offset $T \geq 4$.
Each node $v$ has an internal clock $\delta(v) \in \{0,\ldots,T-1\}$, which is set to $\delta(v) \coloneq \delta(v) + 1 \mod T$ in each round.
If required by the algorithm, a node $v$ is allowed to reset $\delta(v)$ to any desired value.
Initially the $\delta(v)$'s contain arbitrary values out of $\{0,\ldots,T-1\}$.

Communication between nodes is limited to the following rules: In a single round, a node may choose to either \emph{beep} or \emph{listen}.
A beeping node sends a beep to all of its neighboring nodes.
A node $v$ that listens can either decide whether at least one of its neighbors beeped in the same round, or if no neighbor beeped.
Particularly, we assume that nodes are not aware of their neighbors, so $v$ cannot count the number of neighboring nodes that beeped, i.e., it can only decide between 'beep' or 'silence'.
Regardless of the states described above, nodes may perform internal computation in every round.

In order to model \emph{arbitrary activations} of nodes, we assume that initially all nodes are in an \emph{inactive} state.
An inactive node only checks for a beep and is not allowed to perform internal computation.
A node becomes \emph{active} once it either has been activated by an adversary at an arbitrary point in time or once it has heard a beep from a neighboring node.

In the \emph{synchronization problem} (with period $T$) nodes are required to synchronize their clocks such that all nodes become active and beep in the same round every $T$ rounds.
Formally, it should hold that from a certain point in time on, all $\delta(v)$ have the same value in every round, so all nodes beep at the same round (whenever they have clock values $\delta(v)=0$).
It is easy to see that solving this problem requires a worst-case time of at least $\Omega(D)$ rounds since at least $D$ rounds are necessary in order for two nodes that are $D$  hops away from each other to communicate.
Also, every node has to maintain at least $\lceil \log T \rceil$ bits in order for its clock to be able to count to $T$.

\begin{fact} \label{fact:lower_bound}
	Let $G=(V,E)$ be a connected graph of diameter $D$.
	Any distributed protocol that solves the synchronization problem with period $T$ on $G$ needs $\Omega(D)$ rounds in the worst case using $\Omega(\log T)$ bits at each node.
\end{fact}

We say that a protocol for the synchronization problem is \emph{self-stabilizing}~\cite{DBLP:journals/cacm/Dijkstra74} if, starting from an \emph{arbitrary} state (with arbitrary clock values and arbitrary assignments to node variables), the system is guaranteed to reach a legitimate state (\emph{convergence}).
Furthermore, once the system is in a legitimate state, it remains in legitimate states thereafter (\emph{closure}).
A legitimate state in our setting is a state in which all nodes already have synchronized clock values and there are no corrupted protocol variables (we define the latter property more formally at a later point).

\subsection{Related Work}

The discrete beeping model was introduced by Cornejo and Kuhn \cite{DBLP:conf/wdag/CornejoK10} who also presented an algorithm for interval coloring: Given a set of resources, the algorithm assigns a fraction of the resources to every node such that no neighboring nodes share resources.
Their algorithm runs in $O(\log n)$ rounds, which matches the lower bound for that problem.
Other problems that have been considered in this model are maximal independent set \cite{DBLP:journals/dc/AfekABCHK13, DBLP:conf/podc/ScottJ013}, leader election \cite{DBLP:conf/wdag/ForsterSW14, DBLP:conf/wdag/GilbertN15} or rendevouz of two agents \cite{DBLP:journals/ijfcs/ElouasbiP17}.

There already has been some work on synchronization in different variants of the beeping model motivated by fireflies: Gouda and Herman~\cite{DBLP:journals/ipl/GoudaH90} present a self-stabilizing algorithm under the assumption that nodes are aware of their neighbors.
In \cite{DBLP:conf/wdag/GuerraouiM15} the authors consider a variant of the model which allows nodes to count the number of beeps that occurred in a single round and give a self-stabilizing algorithm that works even under the existence of a bounded number of $f$ Byzantine processes in case the communication network is a clique.
Similar to this is the work of Dolev et. al.~\cite{DBLP:journals/jcss/DolevHJKLRSW16} which gives self-stabilizing algorithms for the \emph{synchronous 2-counting} problem (a special case for the synchronization problem with $T=2$) that also work under byzantine failure.

The paper closest to ours is is the work of Alistarh et. al. \cite{firefly_sync} where the authors operate on the very same model\footnote{The authors of~\cite{firefly_sync} divide time into slots, where slot boundaries are synchronized across nodes, which is technically equivalent to having synchronous rounds.} and problem as we do and give an algorithm for the synchronization problem that runs in $O(T \cdot D)$ rounds.
However, they require $T$ to be at least $4n$, so their runtime is also no better than $O(n \cdot D)$, which is inefficient.
Also, their algorithm is not self-stabilizing and requires at least $T \geq 4n$ bits of memory for each node, which is much more compared to our algorithm.

Our synchronization problem can be related to clock synchronization algorithms, which have already received much attention by the community (for some results, see for example \cite{DBLP:journals/dc/FanL06,DBLP:journals/ton/LenzenSW15, DBLP:conf/ipsn/SommerW09, DBLP:conf/sensys/LenzenSW09, DBLP:journals/jacm/LenzenLW10, Romer:2001:TSA:501416.501440}).
Other synchronization protocols specifically motivated by wireless networks can be found, for example, in \cite{DBLP:conf/sensys/LucarelliW04, 4607217}.
While these algorithms operate on more realistic but also much more complex models, they cannot simply be emulated in the beeping model without increasing the runtime for synchronization, as nodes can only broadcast\footnote{Recall that nodes in our model are not aware of their neighbors, while in standard clock synchronization algorithms nodes know their neighbors by their identifier.} one bit to all of its neighbors in a single round here.
Another advantage of our model is that the algorithms are very easy to implement, whereas the above mentioned clock synchronization protocols are extremely complex, require more local space at each node and send more than one bit per message in a single round, which would result in an increased runtime if we were to implement these algorithms in our model.

\subsection{Our Contribution}
Our first result involves a fast algorithm that solves the synchronization problem in an asymptotically optimal way:

\begin{theorem} \label{theorem:super_fast}
	There is a distributed protocol that solves the synchronization problem in at most $$4D + \Bigl\lfloor \frac{D}{\lfloor T/4 \rfloor} \Bigr \rfloor \cdot (T \mod 4) = O(D)$$ rounds for any connected input graph $G=(V,E)$ and any fixed $T \geq 4$.
	The protocol uses at most $\lceil \log T \rceil + O(1)$ bits at each node.
\end{theorem}

This time is asymptotically optimal with regard to the lower bound of $\Omega(D)$ (\Cref{fact:lower_bound}).
Note that the exact number of rounds is no larger than $7D$ in the worst case (set $T=7$) and no lower than $4D$ in the best case (set $T$ to any multiple of $4$).
We also show that our analysis is tight by giving an example in which exactly $7D$ rounds are needed for the system to synchronize.
Nodes only have to maintain $\lceil \log T \rceil + O(1)$ bits in their internal storage, i.e., additional protocol variables other than the clock of a node do not impact the storage capacity for a node asymptotically.
In fact, nodes can store these additional variables using only $3$ bits.

Next we will study protocols for the synchronization problem that are self-stabilizing, which means that the initial configuration of the nodes can be arbitrary.
This means that nodes can already be active and have arbitrary assignments of values to not only their clock value $\delta$, but also to all other protocol-specific variables (except for the constant $T$).
We can prove that the price one has to pay in order for the algorithm to be self-stabilizing is quite high compared to our fast algorithm:

\begin{theorem} \label{theorem:lower_bound}
	Any self-stabilizing protocol $A$ for the synchronization problem that works on any connected graph $G=(V,E)$ with at most $n$ nodes requires each node to have at least $\max\{T, n\}$ states and, furthermore, synchronization will take at least $\max\{T, n\}$ rounds in the worst case.
\end{theorem}

On the positive side, we present a self-stabilizing protocol that solves the synchronization problem in an asymptotically optimal way with regards to both, runtime and space requirements.

\begin{theorem} \label{theorem:self_stabilizing}
	There is a self-stabilizing distributed protocol that solves the synchronization problem in at most $O(\max\{T,n\})$ rounds for any connected input graph $G=(V,E)$ using at most $O(\log(\max\{T,n\}))$ bits at each node.
\end{theorem}

Note that by having at most $O(\log(\max\{T,n\}))$ many bits, each node can have at most $O(\max\{T,n\})$ many states, so our algorithm is asymptotically optimal with regards to space as well.
To the best of our knowledge, this is the first self-stabilizing protocol for the synchronization problem in the beeping model that works on \emph{arbitrary} (connected) networks.

We briefly argue in \Cref{sec:async_model} that our algorithms also work in a continuous time model with slots, whose boundaries are not necessarily in sync among the nodes initially.

\section{Our Algorithms and Techniques, in a Nutshell}

\paragraph{Fast Synchronization}
For our fast synchronization protocol (\Cref{sec:fast_sync}) we first compute $\lfloor T/4 \rfloor$ \emph{checkpoints} from the set of all clock values $\{0,\ldots,T-1\}$:
The idea of our algorithm is that while not all nodes are in sync yet, there exists at least one node that beeps after its clock reaches a checkpoint.
A node $v$ can get \emph{induced} by its neighboring nodes in case it listens to a beep from them if and only if its clock value is right before a checkpoint, i.e., if $\delta(v) = c-1 \mod T$ for some checkpoint $c$.
If $v$ gets induced, it raises its clock value by $2$ instead of $1$.
As we will see, this will cause it to get in sync with all of its neighbors by which it got induced and beeps in the round thereafter to 'spread' the beep among further nodes in the graph.
If $\delta(v)$ is already larger than the value for the largest checkpoint, we choose $0$ to be the next checkpoint, i.e., we set $\delta(v)=1$.
Once all nodes are in sync, we can show that all nodes exclusively beep rounds where the $\delta$-values are $0$.
In fact, a node $v$ with $\delta(v)=0$ at the beginning of a round always beeps, regardless of the state of the whole system.

A key insight of the analysis (\Cref{sec:analysis}) is that all nodes synchronize their clock value to the clock value of the nodes that got activated by the adversary first.
Let the nodes that got activated by the adversary first be denoted by the set $V_A^0$.
We can assign a level to each node which indicates the minimum number of hops in the graph in order to reach one of the nodes in $V_A^0$.
Ultimately, we prove that within every $4$ rounds (which we define as a \emph{period}) all nodes at a certain level are getting in sync with the nodes in $V_A^0$.
As there are at most $D$ levels in the graph, we can conclude that all nodes are in sync after $D$ periods.
There is a slight overhead of $\Bigl \lfloor \frac{D}{\lfloor T/4 \rfloor} \Bigr \rfloor \cdot (T \mod 4)$ rounds in the runtime for cases where $T$ is not a multiple of $4$.
The reason for this is that every $\lfloor T/4 \rfloor^{\mathit{th}}$ period consists of $4+ (T \mod 4)$ rounds because of the way we defined the set of checkpoints $\mathit{CP}$.

\paragraph{A Lower Bound for Self-stabilizing Protocols}
We assume that a distributed protocol executed at a node is a finite state machine and all nodes perform the same protocol.
In order to prove the lower bound on the number of states (\Cref{sec:lower_bound}) we provide specific topologies with initial configurations on which no self-stabilizing protocol can synchronize if the state machine has only a limited number of states available.
Similarly we can also prove a lower bound for the runtime of any self-stabilizing protocol.

\paragraph{Self-stabilizing Synchronization}
In \Cref{sec:self_stab_desc} we extend the fast algorithm from \Cref{sec:fast_sync} such that once a node detects an abnormal behavior, it switches to a $\pulse$ state where it beeps for $4$ consecutive rounds.
This triggers other nodes to switch into the $\pulse$ state as well, corresponding to a 'reset' of the system (which is a fairly standard approach for self-stabilizing systems).
Nodes that have beeped for $4$ consecutive rounds in the $\pulse$ state switch to a $\lock$ state, where they just wait for $4n$ rounds.
Afterwards they switch to the $\inactive$ state.
See \Cref{fig:idea_sf} for an illustration of these states.

\begin{figure*}[ht]
	\centering
	\includegraphics[scale=0.9]{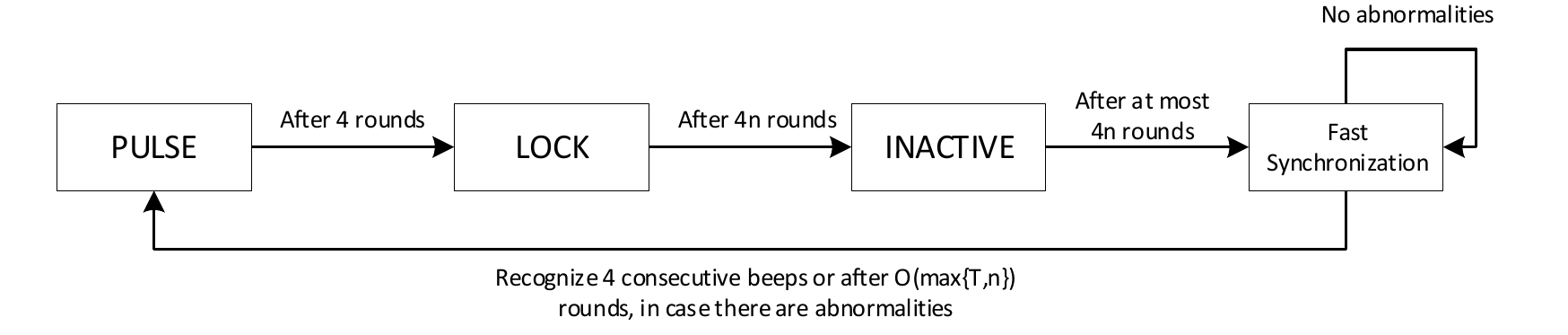}
	\caption{Different states in our self-stabilizing protocol from a high-level perspective.}
	\label{fig:idea_sf}
\end{figure*}

We can show that when letting the nodes start from any initial state, then within at most $O(\max\{T,n\})$ rounds we reach a round where all nodes are in state $\lock$.
Once this happens, we just have to wait for another $O(n)$ rounds until all nodes are in state $\inactive$.
From that point on our algorithm behaves exactly as the fast algorithm, yielding convergence and closure after another $O(D)$ rounds.

\section{Fast Synchronization} \label{sec:fast_sync}
\Cref{table:variables} states the variables that need to be stored by each node $v \in V$.
Note that $v$ can maintain $\delta(v)$ using $\lceil \log T \rceil$ bits, while only using $2$ bits for $\s(v)$ and $1$ bit for $\induced(v)$.

\begin{table}[ht]
\centering
\ra{1.3}
\begin{tabular}{@{}lp{14cm}@{}}
\toprule 
 $\delta(v)$ & A counter out of $\{0,\ldots,T-1\}$ simulating the internal clock of node $v$. Initially this variable contains an arbitrary value out of $\{0,\ldots,T-1\}$.\\
 $\s(v)$ & A flag out of $\{\inactive, \beep, \listen\}$ used to indicate the state of $v$. \\
 $\induced(v)$ & A flag out of $\{ \mathit{true}, \mathit{false}\}$ indicating whether $v$ got induced recently by another node or not.\\
\bottomrule
\end{tabular}
\caption{Variables used by each node $v \in V$}
\label{table:variables}
\end{table}

We define the set of checkpoints as follows:

\begin{definition} \label{def:checkpoint}
	Let $T \geq 4$ be a fixed integer.
	Define the set of \emph{checkpoints} $\mathit{CP} = \{c \in \mathbb{N}\ |\ c \mod 4 = 0\ \wedge\ T-c > 3 \}$.
	A value $c \in \mathit{CP}$ is called a \emph{checkpoint}.
\end{definition}

For example, if $T=19$, we have checkpoints $\{0,4,8,12\}$.
Note that by definition, $16$ is not a checkpoint in this case.
This has to hold, because we want the distance between two checkpoints (modulo $T$) to be at least $4$, so since $16+3 \equiv 0 \mod T$, the distance between $16$ and $0$ is too small.
Since all nodes know the same value for $T$, they are able to compute $\mathit{CP}$ by themselves.

\begin{algorithm}[ht]
\caption{Pseudocode executed at each node $v$ in each round}
\label{algo:protocol}
\begin{algorithmic}[1]
	\If{$\s(v) = \inactive$}  \label{algo:inactive_start}
		\If{$\exists w \in N(v): w$ beeps or the adversary wakes up $v$}
			\State $\delta(v) \gets 1$ \label{algo:inactive_middle}
			\State $\s(v) \gets \beep$ 
			\State $\induced(v) \gets true$ \label{algo:inactive_end}
		\EndIf
	\ElsIf{$\s(v) = \beep$} \Comment{Check if $v$ should beep} \label{algo:check:bnr_true}
		\State Beep!
		\State $\delta(v) \gets \delta(v) + 1 \mod T$ \label{algo:delta_incr:0}
		\State $\s(v) \gets \listen$ \label{algo:delta_incr:1}
	\ElsIf{$\exists w \in N(v): \s(w)=\beep$} \Comment{Listen for a beep}
		\If{$\delta(v) = c - 1 \mod T$ for some $c \in \mathit{CP}$} \label{algo:check_checkpoint}
			\State $\delta(v) \gets \delta(v)+2 \mod T$ \label{algo:induced_start}
			\State $\s(v) \gets \beep$ \Comment{Induced beep in the next round} \label{algo:bnr_true:1}
			\State $\induced(v) \gets true$ \label{algo:induced_true:1}
		\Else \label{algo:checkpoint}
			\State $\delta(v) \gets \delta(v) + 1 \mod T$ \label{algo:no_induction}
		\EndIf
	\Else \Comment{No beep occurred} \label{algo:no_noise_start}
		\State $\delta(v) \gets \delta(v) + 1 \mod T$ \label{algo:increment_delta}
		\If{$\left( \induced(v) \textbf{ and } \delta(v) \in \mathit{CP} \right)$ \textbf{or} $\delta(v) = 0$} \label{algo:check_checkpoint_2}
			\State $\s(v) \gets \beep$ \Comment{Mature beep in the next round} \label{algo:bnr_true:2}
			\State $\induced(v) \gets false$ \label{algo:reset_induced}
		\EndIf
	\EndIf
\end{algorithmic}
\end{algorithm}

\Cref{algo:protocol} states the actions that are performed by every node $v$ in each round.
In general, nodes perform actions based on the state they are in at the beginning of the round: If an inactive node $v$ gets activated (either by the adversary or by one of its neighbors) it sets $\delta(v) \gets 1$, $\induced \gets true$ and $\s(v) \gets \beep$ (Lines~\ref{algo:inactive_start} to \ref{algo:inactive_end}).
A node $v$ with $\s(v)=\beep$ beeps, increments its clock value by $1$ and switches to state $\listen$ (\Crefrange{algo:check:bnr_true}{algo:delta_incr:1}).

A node $v$ that is in state $\listen$ checks if any neighbor $w$ beeps.
We first describe the case where at least one neighbor $w$ of $v$ beeps: $v$ first checks if its $\delta$-value is one less than the next checkpoint, i.e., if there exists a checkpoint $c \in \mathit{CP}$ such that $c-1 \mod T = \delta(v)$.
If there is such a $c$, then $v$ sets $\delta(v)$ to $\delta(v) + 2 \mod T$, $\s(v)$ to $\beep$ and the flag $\induced(v)$ to $true$ (\Crefrange{algo:check_checkpoint}{algo:induced_true:1}).
In this case we say that $v$ got \emph{induced} by $w$, i.e., $v$'s beep in the next round is an \emph{induced beep}.
If the above condition does not hold, then $v$ just increments $\delta(v)$ by $1$ (mod $T$) (\Cref{algo:checkpoint}).

Now we state $v$'s action, in cases where $\s(v) = \listen$, but no beep from any of $v$'s neighbors occurred (\Cref{algo:no_noise_start}): Here $v$ first increments $\delta(v)$ by $1$ (mod $T$, \Cref{algo:increment_delta}) and then checks, whether its flag $\induced(v)$ is $\mathit{true}$ and if $\delta(v)$ is a checkpoint (\Cref{algo:check_checkpoint_2}).
In this case, $v$ sets $\s(v) = \beep$ (\Cref{algo:bnr_true:2}) and resets $\induced(v)$ to $\mathit{false}$ (\Cref{algo:reset_induced}).
This is also done if $\delta(v)=0$, because we always want nodes with $\delta(v)=0$ to beep once the system has fully synchronized.
We say that $v$ beeps \emph{maturely} in this case.

\section{Analysis of the Fast Synchronization Protocol} \label{sec:analysis}

We start by introducing some notation.
Denote by $\s_t(v)$ the value of $\s(v)$ at the beginning of round $t$ and by $\delta_t(v)$ the value of $\delta(v)$ at the beginning of round $t$.
We say that two nodes $v,w \in V$ are \emph{in sync} in round $t$ if $\delta_t(v) = \delta_t(w)$.
To ease notation we consider rounds $0,1,2,\ldots$, where $0$ is the first round in which any node is active for the first time.
This implies that there is at least one node that already beeps in round $0$.

For each node $v \in V$ we assign a virtual counter $c_t(v) \in \mathbb{N}_0$ to $v$ which indicates the absolute value by which $\delta(v)$ has increased overall until round $t$, where $c_t(v)$ is initialized to $0$ if $v$ is activated in round $t$.
We can show the following for the relation between the virtual counters and clock values:

\begin{lemma} \label{lemma:virtual_counter:clock_value}
	For any active node $v$ it holds in any round $t$ that $\delta_t(v) = 1 + c_t(v) \mod T$.
\end{lemma}

\begin{proof}
	Follows from the fact that we only increase $\delta(v)$ within the ring $\mathbb{Z}/ T\mathbb{Z}$ in our algorithm and that each activated node starts with the same clock value of $1$ in our protocol.
\end{proof}

The following statement is a direct implication of \Cref{lemma:virtual_counter:clock_value}:

\begin{corollary} \label{cor:counter_clock_relation}
	Let $v,w \in V$ be two active nodes in round $t$.
	If $c_t(v) = c_t(w) + x$, then $\delta_t(v) = \delta_t(w) + x \mod T$.
\end{corollary}

Note that the converse is not true since, for example, for $\delta_t(v) = \delta_t(w) = 1$ it may hold $c_t(v) = T > 0 = c_t(w)$.

We now make some simple claims that follow directly from the description of our algorithm:

\begin{lemma} \label{lemma:beeping_behavior}
	The following statements hold for any active node $v \in V$.
	\begin{itemize}
		\item[(i)] If $v$ beeps maturely in round $t$, then $\delta_t(v) \in \mathit{CP}$.
		\item[(ii)] If $v$ beeps induced in round $t$, then $\delta_t(v) = c+1 \mod T$ for some $c \in \mathit{CP}$.
		\item[(iii)] If $v$ gets induced in round $t$, then $\delta_t(v) = c - 1 \mod T$ for some $c \in \mathit{CP}$.
		\item[(iv)] If $\s_t(v) = \beep$ for some round $t$, then either $\delta_t(v) \in \mathit{CP}$ or $\delta_t(v) = c+1 \mod T$ for some $c \in \mathit{CP}$.
	\end{itemize}
\end{lemma}

Using these claims we can prove the following:

\begin{lemma} \label{lemma:no_induce_possible}
	Let $v,w \in V$ be active neighboring nodes in round $t$.
	\begin{itemize}
		\item[(i)] If $c_t(v) = c_t(w)+1$, then $w$ cannot induce $v$ in round $t$.
		\item[(ii)] If $c_t(v) = c_t(w)$, then $w$ cannot induce $v$ in round $t$.
	\end{itemize}
	
\end{lemma}

\begin{proof}
	For (i) assume to the contrary that $w$ induces $v$ in round $t$.
	We consider two cases.
	In the first case $w$'s beep that induces $v$ is mature.
	By \Cref{lemma:beeping_behavior}(i) it follows that $\delta_t(w) = c$ for some checkpoint $c \in \mathit{CP}$.
	But then it holds by \Cref{lemma:beeping_behavior}(iii) that $\delta_t(v) = c-1$ which contradicts \Cref{cor:counter_clock_relation} because the difference between two checkpoints is always strictly greater than $2$.
	Now assume that $w$'s beep in round $t$ is an induced beep.
	Then by \Cref{lemma:beeping_behavior}(ii) it holds that $\delta_t(w) = c+1$ for some checkpoint $c \in \mathit{CP}$.
	But then it holds by \Cref{lemma:beeping_behavior}(iii) that $\delta_t(v) = c-1$ which again contradicts \Cref{cor:counter_clock_relation} this time because the difference between two checkpoints is always strictly greater than $3$.
	
	For property (ii) it has to hold that $\s_t(w) = \beep$ and thus by \Cref{lemma:beeping_behavior}(iv) either $\delta_t(w) \in \mathit{CP}$ or $\delta_t(w) = c+1$ for some $c \in \mathit{CP}$.
	In order for $v$ to get induced by $w$ it has to hold $\s_t(v) = \listen$ and $\delta_t(v) = c'-1 \mod T$ for some checkpoint $c' \in \mathit{CP}$.
	However, since $c_t(v) = c_t(w)$ it follows from \Cref{cor:counter_clock_relation} that $\delta_t(v) = \delta_t(w)$, so it has to hold $c'-1 = c$ or $c'-1 = c+1$.
	This only holds if $c' \neq c$, but then we contradict the fact that $c_t(v) = c_t(w)$, thus (ii) holds.
\end{proof}

The following observation can easily be checked in the pseudocode of \Cref{algo:protocol}:

\begin{observation} \label{obs:counter_increase_bounds}
	Let $v \in V$ be an active node in round $t$.
	Then it holds $c_t(v) + 1 \leq c_{t+1}(v) \leq c_t(v)+2$, i.e., the virtual counter of $v$ increases by at least $1$ and by at most $2$ in each round.
\end{observation}

From \Cref{lemma:no_induce_possible}(ii) and \Cref{obs:counter_increase_bounds} we can conclude that the node $v$ with highest value $c_t(v)$ at any round $t$ keeps having the highest virtual counter for the remaining time.
For this denote by the set $V_A^i \subseteq V$ the nodes that are activated by the adversary in round $i$.
Due to \Cref{obs:counter_increase_bounds} it is clear that any node $v \not \in V_A^0$ that is activated in round $t$ has a virtual counter no higher than nodes in $V_A^0$ in round $t$.

\begin{corollary} \label{cor:highest_counter_remains}
	Let $v \in V_A^0$ be an active node in round $t$ with $c_t(v) \geq c_t(w)$ for any active node $w \neq v$.
	Then $c_{t'}(v) \geq c_{t'}(w)$ still holds for any round $t'>t$.
\end{corollary}

\Cref{cor:highest_counter_remains} implies that in case all nodes synchronize, they synchronize their clock values to the clock values of nodes in $V_A^0$.
We still need to show that this is also what actually happens in our system.

Let $d_t(v,w) = |c_t(v)-c_t(w)|$ denote the absolute difference in the values of $v$'s and $w$'s virtual counters in round $t$.

One can easily observe the following:

\begin{observation} \label{obs:difference_chain}
	Let $v \in V$ be a node that gets activated in round $t$ and let $w \in V$ be a neighboring node of $v$ that got activated in round $t-1$.
	Then $d_t(v,w) = 1$.
\end{observation}

The next technical lemma states that the differences of virtual counters for neighboring nodes never gets larger than $1$ for two or more consecutive rounds.

\begin{lemma} \label{lemma:reset}
	Let $v,w \in V$ be two active neighboring nodes in round $t$ with $d_t(v,w) \leq 1$.
	Let $t' > t$ be the first round where $d_{t'}(v,w) > 1$ holds.
	Then $d_{t'+1}(v,w) \leq 1$ again.
\end{lemma}

\begin{proof}
	W.l.o.g. let $d_t(v,w) = 1$, $c_t(v) = c_t(w) + 1$ and let $t' = t+1$ with $d_{t'}(v,w) > 1$.
	Due to \Cref{obs:counter_increase_bounds} it has to hold that $d_{t'}(v,w) = 2$.
	Assume to the contrary that $d_{t'+1}(v,w) > 1$ still holds.
	As $d_t(v,w)$ increased by $1$ when going from round $t$ to $t'$ it follows by our assumptions and \Cref{obs:counter_increase_bounds} that $c_{t'}(v) = c_t(v) + 2$ and $c_{t'}(w) = c_t(w) + 1$.
	By our algorithm the only way for $v$ to increase its virtual counter by $2$ within one round is getting induced by some other node $v'$ ($v$ cannot get induced by $w$ here due to \Cref{lemma:no_induce_possible}).
	Thus it holds by \Cref{lemma:beeping_behavior}(iii) that $\delta_t(v) = c - 1 \mod T$ for some checkpoint $c \in \mathit{CP}$.
	Since $c_t(v) = c_t(w) + 1$ it holds by \Cref{cor:counter_clock_relation} that $\delta_t(w) = c - 2 \mod T$.
	Let us now compute the clock values of $v$ and $w$ in round $t'$: As $v$ gets induced it sets $\delta_{t'}(v) = \delta_t(v) + 2 = c + 1 \mod T$.
	Node $w$ just increases its clock by $1$, so it sets $\delta_{t'}(w) = \delta_t(w) + 1 = c - 1 \mod T$.
	As $v$ got induced in round $t$, it performs an induced beep in round $t'$, i.e., $\s_{t'}(v) = \beep$.
	By \Cref{lemma:beeping_behavior}(iv) it has to hold that $s_{t'}(w) = \listen$.
	Since $\delta_{t'}(w) = c-1 \mod T$, $w$ gets induced by $v$ and increments its clock by $2$, i.e., it sets $\delta_{t'+1}(w) = \delta_{t'}(w) + 2 = c + 1 \mod T$.
	The fact that $v$ beeps in round $t'$ implies that it increments its clock by $1$, so it sets $\delta_{t'+1}(v) = \delta_{t'}(v) + 1 = c + 2 \mod T$.
	By the definition of the virtual counters we immediately get $d_{t'+1}(v,w) = 1$ and arrive at a contradiction.
\end{proof}

%

For convenience, denote the \emph{configuration} in round $t$ for a node $v \in V$ by $C_t(v)=(\delta_t(v),$ $\s_t(v),$ $\induced_t(v))$.
For a value $x \in \{0,\ldots,T-1\}$, we denote by 
  \begin{equation*}
    \mathit{succ}(x) =
    \begin{cases*}
      \min(\{c \in \mathit{CP} \mid c > x\}) & if $\exists c \in \mathit{CP}: c > x$ \\
      0        & otherwise
    \end{cases*}
  \end{equation*}
the \emph{successor} of $x$, i.e., the smallest checkpoint larger than $x$, or $0$ if $x$ is larger than any checkpoint.

We partition the rounds $1,2,\ldots, 4D + \Bigl\lfloor \frac{D}{\lfloor T/4 \rfloor} \Bigr \rfloor \cdot (T \mod 4)$ into contiguous intervals (which we call \emph{periods}) $P_1,\ldots,P_D$, where each $P_i$ consists of two subsequent checkpoints $c, \mathit{succ}(c)$ and contains exactly $\mathit{succ}(c) - c \mod T$ many rounds.
For example, if $T = 19$, then we have $\mathit{CP} = \{0,4,8,12\}$ and consequently $P_1 = \{1,2,3,4\}$ (for checkpoints $4$ and $0$), $P_2 = \{5,6,7,8\}$ (for $8$ and $4$), $P_3 = \{9,10,11,12\}$ (for $12$ and $8$) and $P_4 = \{13,14,15,16,17,18,19\}$ (for checkpoints $0$ and $12$ - note that $|P_4| = 4 + (T \mod 4) = 7$ because $0-12 \mod T = 7$).
The period $P_5$ then starts again with checkpoints $4$ and $0$, i.e., $P_5 = \{20,21,22,23\}$ and so on.
We are now ready to prove the main result of this section.

\begin{proof}[Proof of \Cref{theorem:super_fast}]
	Let us define the set $S \subseteq V$ as the set of nodes that have their clock values in sync with the clock values of nodes in $V_A^0$.
	At round $0$ it holds $S = V_A^0$.
	We show via induction over all periods that in every period $P$ all nodes $v \not \in S$ that are neighbors of at least one node $w \in S$ will be added to $S$ at the beginning of the last round of $P$. 
	For this we make the following claims:
	\begin{itemize}
		\item[(i)] In the first round $t$ of $P$, it holds for $\{v,w\} \in E$ with $v \not \in S$ and $w \in S$ that $C_t(w) = (c+2,\listen,\mathit{true})$ and $C_t(v) = (c+1,\beep,\mathit{true})$, where $c \in \mathit{CP}$ is a checkpoint.
		\item[(ii)] In the last round $t'$ of $P$, it holds that $C_{t'}(w) = (c'+1,$ $\listen, \mathit{false})$ and $C_{t'}(v) = (c'+1,\beep,\mathit{true})$ where $c' = \mathit{succ}(c)$.
		This implies that $v$ gets added to $S$ in round $t'$.
	\end{itemize}		
	
	For the base case consider the period $P_1 = \{1,2,3,4\}$.
	Since $w$ got activated in round $0$, $v$ gets activated in round $1$ and thus, by \Cref{obs:difference_chain} it holds $d_1(v,w) = 1$.
	
	Also one can easily verify that in round $1$ we have $C_1(w) = (2,\listen,\mathit{true})$ and $C_1(v) = (1,\beep,\mathit{true})$, so the claim (i) holds.
	Due to our algorithm, now both nodes just increment their clocks by $1$, while $v$ also switches to state $\listen$.
	Therefore, we get $C_2(w) = (3,\listen,\mathit{true})$ and $C_2(v) = (2,\listen,\mathit{true})$
	Again, due to our algorithm, both nodes just increment their clocks by $1$.
	Note that, due to \Cref{lemma:no_induce_possible}(i), $w$ does not get induced in this round.
	As the clock value for $w$ is set to $4 \in \mathit{CP}$ this way and $\induced_2(w) = \mathit{true}$ holds, it triggers \Cref{algo:check_checkpoint_2} of \Cref{algo:protocol}, so $w$ switches its state to $\beep$ for the next round and sets its $\induced$ flag to $\mathit{false}$.
	Hence, we get for round $3$ that $C_3(w) = (4,\beep,\mathit{false})$ and $C_3(v) = (3,\listen,\mathit{true})$ 
	Now, all conditions for $v$ to get induced are met, so it increments its counter by $2$ in round $3$.
	Hence, $v$ gets in sync with $w$ at the beginning of round $4$.
	Ultimately, we get $C_4(w) = (5,\listen,\mathit{false})$ and $C_4(v) = (5,\beep,\mathit{true})$, so the claim (ii) holds.
	
	For the induction step assume that the claims hold for period $P_i$.
	We now argue that the claims hold for the period $P_{i+1}$ as well.
	Let $S' \subset S$ be the set of nodes that got added to $S$ in the last round of period $P_i$.
	It has to hold that $S' \neq \emptyset$, otherwise the system would have been in sync already (this is due to the fact that $S'$ is the immediate neighborhood of $S \setminus S'$).
	By the same argument there has to exist a node $w \in S'$ such that there is an edge $\{v,w\} \in E$ with $v \not \in S$.
	Due to the induction hypothesis it holds that in the last round $t$ of $P_i$ we have $C_t(w) = (c+1,\beep,\mathit{true})$ for a checkpoint $c \in \mathit{CP}$.
	Thus we have $C_{t+1}(w) = (c+2,\listen,\mathit{true})$ in the first round of period $P_{i+1}$.
	As $1 \leq d_{t+1}(v,w) \leq 2$ we get that either $C_{t+1}(v) = (c+1,\beep,\mathit{true})$ or $\delta_{t+1}(v) = c$.
	Note that $c(v)$ cannot be larger than $c(w)$ because $w$ already got included into the set $S$, so its counter $c(w)$ is already the same as all nodes in $V_A^0$ and thus maximal.
	The latter case does not happen due to the following reason: Since $d_{t+1}(v,w) = 2$, it has to hold due to \Cref{lemma:reset} that $\delta_{t+1}(v)$ increases by $2$ in round $t+1$.
	This implies that $d_{t+2}(v,w)$ gets back to $1$.
	This would only be possible if $v$ gets induced in round $t+1$, but this cannot happen as $\delta_{t+1}(v)=c$ is a checkpoint.
	Thus we only need to consider the case where $C_{t+1}(v) = (c+1,\beep,\mathit{true})$.
	From this point onward the same arguments hold as in for the base case (with the exception that we have to consider arbitrary checkpoints $c,\mathit{succ}(c) \in \mathit{CP}$ instead of concrete checkpoints $0$ and $4$).
	This concludes the induction.
	
	The above induction shows that we add nodes to $S$ in a BFS fashion starting at all nodes in $V_A^0$.
	This implies that we need exactly $D$ periods until all nodes are in $S$.
	Each period consists of exactly $4$ rounds, except the period that considers checkpoints $c = \max\{c \in \mathit{CP}\}$ and $c' = 0$ which consists of $l = 4 + (T \mod 4)$ rounds.
	Within $D$ periods a period of $l$ rounds is considered at most $\Bigl\lfloor \frac{D}{\lfloor T/4 \rfloor} \Bigr \rfloor$ times so we get the overhead of $T \mod 4$ rounds for this amount of times, which proves the runtime of $4D + \Bigl\lfloor \frac{D}{\lfloor T/4 \rfloor} \Bigr \rfloor \cdot (T \mod 4)$ for our protocol.
\end{proof}

\subsection{Tightness}
We show that our analysis is tight (recall that the runtime stated in \Cref{theorem:super_fast} cannot get larger than $7D$).
Consider the following example depicted in \Cref{fig:wc_example}(a) with $D=3$ and assume that $T = 7$.
Note that for $T=7$ it holds $\mathit{CP}=\{0\}$, so $0$ is the only checkpoint.

\begin{figure}[ht]
	\centering
	\includegraphics[scale=0.8]{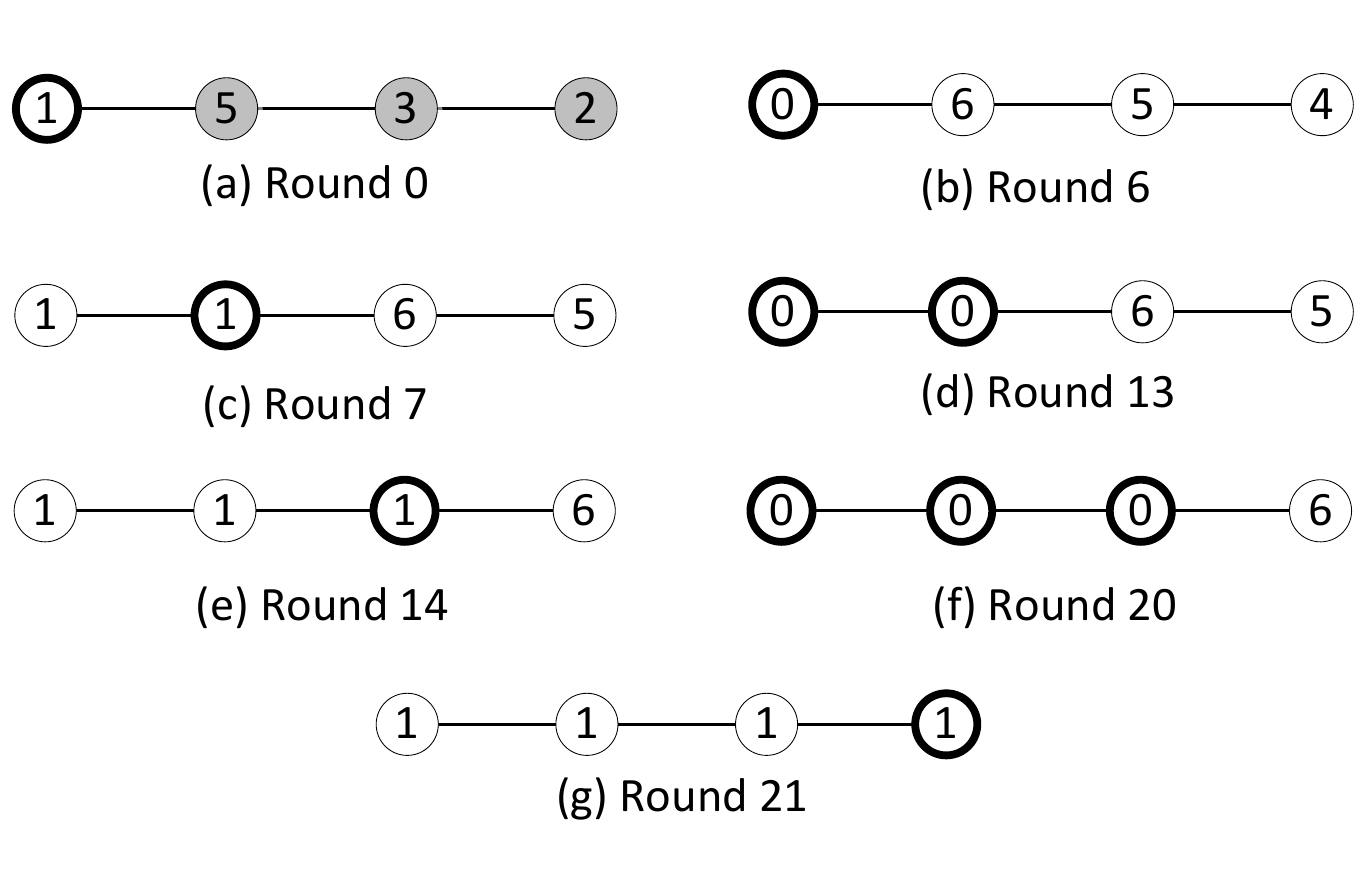}
	\caption{Worst case example that has a runtime of $7 \cdot D$. Node labels represent the clock values $\delta$. Bold circled nodes are beeping, grey nodes are inactive and all the other nodes are listening.}
	\label{fig:wc_example}
\end{figure}

Here the left node has been activated by the adversary first, which lets it activate all the other nodes on the line within the next $3$ rounds.
In round $6$ the outer left node beeps maturely and thus generates an induced beep on its neighbor (\Cref{fig:wc_example}(b)).
Therefore, these two nodes get synchronized in round $7$ (\Cref{fig:wc_example}(c)).
This furthermore triggers a series of induced beeps up until the outer right node of the line (without nodes being synchronized to the two outer left nodes).
After another $7$ rounds the outer left node (and its neighbor) starts to beep again (\Cref{fig:wc_example}(d)), which will get the third node synchronized via an induced beep (\Cref{fig:wc_example}(e)).
Finally, in round $20$ the three left nodes beep maturely again, which forces the last remaining unsynchronized node to get induced (\Cref{fig:wc_example}(f)).
Hence, all nodes are synchronized in round $7 \cdot 3 = 21$, which implies that our analysis is tight.

It is easy to see that this example for the line can be generalized to line topologies of arbitrary diameter such that the overall runtime of our algorithm is exactly $7D$ rounds.


\section{A Lower Bound for Self-stabilizing Protocols} \label{sec:lower_bound}
We show \Cref{theorem:lower_bound} in this section.
By assumption, all nodes are of the same type of finite state machine. In this section we prove an upper bound on the number of nodes of graphs that a synchronization algorithm can solve depending on the number of states its corresponding machine has. 
Equivalently, this gives a lower bound on the number of states of the finite state machine if we hope to solve the synchronization problem on all graphs of size up to $n$. 
We also obtain a lower bound on the synchronization runtime for any self-stabilizing algorithm that works correctly on all graphs of size up to $n$.

\begin{proof}[Proof of \Cref{theorem:lower_bound}]
For the first statement, we show the following: 

\begin{lemma}\label{lemma:first_part_Thm2}
For any algorithm $A$ requiring no more than $k$ states that is executed by every node, there exists a graph of at most $k + 1$ many nodes and a set of initial configurations on which the nodes do not eventually synchronize.
\end{lemma}

\begin{proof}
Consider the directed graph $H$ for the transition diagram of algorithm $A$, where vertices correspond to states and each state $q$ has two outgoing edges for the transitions on hearing a beep and hearing silence: $E(H) = \{ (q, \delta(q, i)) : i \in \{\text{beep, silence}\}, q \in V(H)\}$ where $\delta(q, i)$ denotes the state that a node transitions to from state $q$ upon receiving input $i$. 
Let $P = (p_0, p_1, p_2, \dots)$ be a path traversed in $H$ by a node that always hears beeps, i.e., $p_{i+1} = \delta(p_i, \text{beep})$. Since there are finitely many states, $P$ will have a repeated state which together with the fact that transitions are deterministic implies that $P$ cycles. 
There is therefore a cycle $L = (\ell_0, \ell_1, \dots, \ell_r) \subset V(H)$ where $\ell_r = \ell_0$ and $\delta(\ell_i, \text{beep}) = \ell_{i+1}$ for $i < r$. 
For such a cycle $N$, we distinguish two cases:

\begin{enumerate}[label = \alph*)]
\item $L$ does not contain a state in which the node beeps
\item $L$ contains a state $b$ in which the node beeps.
\end{enumerate}

\textbf{Case a1:}  A node that never hears a beep does not eventually begin beeping once every $T$ rounds for every initial state. 
In this case, there exists an initial configuration such that the graph consisting of a single node  does not synchronize.

\textbf{Case a2:} Case a holds, but not case a1.
In this case there will exist a cycle $M$ of length a multiple of $T$ on which a single lone node is in sync (with itself). If we denote its states by $m_0, m_1, \dots, m_{T-1}, \dots$, not hearing a beep makes a node's state transition from one subscript to the next mod $|M|$, and its clock value corresponds to the subscripts of these states mod $T$. As required by the problem definition, a node on this cycle should beep whenever it has clock value zero, which occurs on states of $M$ with subscripts that are multiples of $T$.  

We construct a star graph of degree $T$, place on each of its leaves a node in such a state for every clock value, and set the center node in any of the states of $L$. See \Cref{fig:graphs}. It is easy to see by induction that the leaves never hear a beep and remain in their cycles undisturbed, while at every time step exactly one leaf node will beep and keep the center node on cycle $L$. 

\textbf{Case b:} 
Here we construct a complete graph of size $|L|$ and set the state of each node to a different state in $L$. It is clear that all nodes remain on $L$ as there is always at least the one node in state $b$ that beeps and is heard by all other nodes. See \Cref{fig:graphs}.

Therefore, if the finite state machines have $k$ states, there exists a graph of at most $\max\{1, T + 1, |L|\}$ nodes and an initial set of states for the nodes such that the nodes never synchronize. 
Since we must have $T \le k$ and $|L| \le k$, this is no more than $k + 1$.

\end{proof}

\begin{figure}[ht]
	\centering
	\includegraphics[scale=1.0]{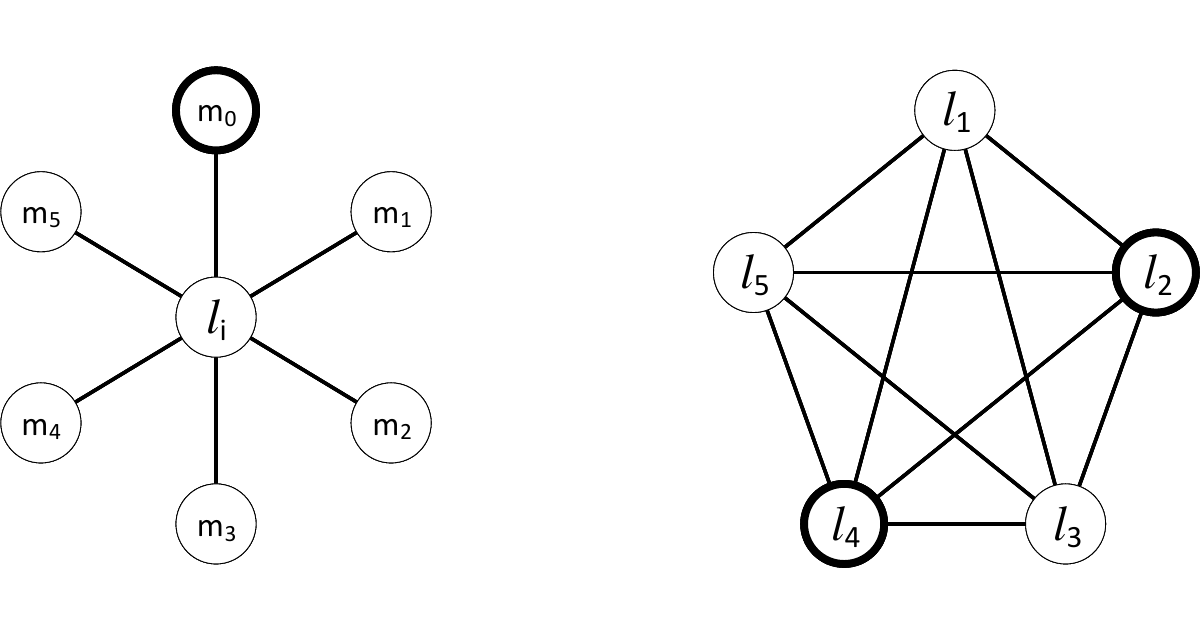}
	\caption{Graph constructions for case (a) on the left with $T = 6$ and case (b) on the right with $|L| = 5$. Bold circled nodes are beeping, all other nodes are listening.}
	\label{fig:graphs}
\end{figure}

This proves that the number of states should be at least $n$. Moreover, after reaching synchronization, there must be at least $T$ many distinct clock values that the nodes can have, thus the number of states used in $A$ must be at least $T$, completing the first statement of \cref{theorem:lower_bound}.

For the second statement, we show the following result:

\begin{lemma}
For any algorithm $A$ that is able to synchronize the nodes of every graph of up to $n$ nodes regardless of topology, there exists a graph of at most $n$ nodes and an initial configuration on which it takes at least $\max\{T, n\}$ rounds for $A$ to synchronize the nodes in the worst case.
\end{lemma}

\begin{proof}
To show that synchronization will take at least $T$ many rounds in the worst case, consider only two nodes connected by a single edge. Following the notation for the cycle $M$ defined above for the single node in sync with itself, initialize one to $m_0$ and the other to $m_1$. After $T - 1$ many rounds, the first will be in state $m_{T-1}$ and the other will either be in $m_{0}$ or $m_T$ based on whether $T = |M|$ or $T < |M|$ and they have not synchronized yet. The earliest time in which they can possibly have synchronized is therefore $T$, because neither node beeped up to this round.

In order to show that synchronization will take at least $n$ rounds in the worst case, we again make the same case distinction as in the proof of the previous lemma based on the cycle $L$ defined therein, and now assuming we have finite state machines that are able to correctly synchronize any graph of up to $n$ nodes, we have for these two cases:

\textbf{Case a:}   
Since $A$ is able to synchronize the nodes, we must have in this case $T + 1 > n$ so that the counterexample of the previous lemma cannot be constructed. This means $T \ge n$, but we have already shown that synchronization takes at least $T$ rounds, concluding this case.

\textbf{Case b:}
Here we must have $|L| = r > n$ for the counterexample of the previous lemma to be impossible to construct. By assumption, cycle $L$ contains a state $b$ in which the node beeps. Number the states of $L$ in order from $0$ to $r-1$ such that b = $\ell_0$. Construct a complete graph of size $n$ and initialize the nodes on states $b = \ell_0, \ell_{r-1}, \dots, \ell_{r-n+1}$. It is now easy to see that for $n$ many rounds the nodes remain on the cycle $L$ and that in this case the system will not have synchronized until after $n$ many rounds.
\end{proof}

This proves that any self-stabilizing protocol takes at least $\max\{n,$ $T\}$ rounds to synchronize. 
This completes the proof of \cref{theorem:lower_bound}.
\end{proof}

\section{Self-stabilizing Synchronization} \label{sec:self_stab_desc}
In this section we describe a self-stabilizing protocol that solves the self-stabilizing synchronization problem in time $O(\max\{T,n\})$, which is asymptotically optimal.
According to \Cref{sec:lower_bound}, an algorithm requres at least $\max\{T, n\}$ many states in order for it to be self-stabilizing.
For this we assume that nodes know the value $n$ (or at least some value $N \geq n$ with $N \in \Theta(n)$).
The lower bound from \Cref{sec:lower_bound} also holds under this assumption as we chose the most general algorithm possible for each node without any restrictions.
First, we want to generalize the notion of checkpoints for values $q > 4$.

\begin{definition} \label{def:checkpoint_ext}
	Let $T \geq 5$ be a fixed integer and $q \in O(1)$ a constant with $5 \leq q \leq T$.
	Define the set of checkpoints w.r.t. $q$ as $\mathit{CP}_q = \{c \in \mathbb{N}\ |\ c \mod q = 0\ \wedge\ T-c > q-1 \}$.
\end{definition}

One can easily verify that the fast algorithm (\Cref{algo:protocol}) still works when generalizing the notion of checkpoints used in \Cref{def:checkpoint} and \Cref{def:checkpoint_ext}.
We get the following corollary on its runtime (follows trivially from \Cref{theorem:super_fast} and \Cref{def:checkpoint_ext}).

\begin{corollary} \label{cor:super_fast_ext}
	Let $T \geq 5$ be a fixed integer and $q \in O(1)$ a constant with $5 \leq q \leq T$.
	There is a distributed protocol that solves the synchronization problem in at most $$qD + \Bigl\lfloor \frac{D}{\lfloor T/q \rfloor} \Bigr \rfloor \cdot (T \mod q) = O(D)$$ rounds for any connected input graph $G=(V,E)$.
\end{corollary}

For ease of presentation we define $\mathit{sf}(q)$ as $$\mathit{sf}(q) = q(n-1) + \Bigl\lfloor \frac{n-1}{\lfloor T/q \rfloor} \Bigr \rfloor \cdot (T \mod q) + q.$$
As any graph has diameter at most $n-1$, $\mathit{sf}(q) \in O(n)$ is an upper bound for the time it takes to solving the synchronization problem in any graph in addition to a small overhead of $q$, whose purpose will become clear later on.

\subsection{Variables}
We extend the set of variables (\Cref{table:self_stabilizing_variables}) for a node $v$ from the fast algorithm by adding two additional states $\{\pulse, \lock\}$ and by introducing two additional counters $r(v)$ and $b(v)$.

\begin{table}[ht]
\centering
\ra{1.3}
\begin{tabular}{@{}lp{14cm}@{}}
\toprule 
	$\delta(v)$ & A counter out of $\{0,\ldots,T-1\}$ simulating the internal clock of node $v$. \\
	$\s(v)$ & A flag out of $\{\inactive$, $\beep$, $\listen$, $\pulse$, $\lock\}$ used to indicate the state of $v$. \\
	$\induced(v)$ & A flag out of $\{ \mathit{true}, \mathit{false}\}$ indicating whether $v$ got induced recently by another node or not.\\
	$r(v)$ & A counter out of $\{0,\ldots,\max\{4n,\mathit{sf}(q)+1\}\}$ for counting the number of rounds that $v$ is in, in a certain state.\\
	$b(v)$ & A counter out of $\{0,\ldots,4\}$ for counting the number of consecutive beeps that $v$ listens to.\\
\bottomrule
\end{tabular}
\caption{Variables used by each node $v \in V$ in the self-stabilizing protocol}
\label{table:self_stabilizing_variables}
\end{table}

The maximum amount of bits that a node needs to store depends on the values for $T$ and $n$, i.e., the variables can be stored by maintaining $O(\max\{\log T, \log n\})$ bits at each node.
As required for self-stabilizing systems, all variables contain arbitrary values out of their respective domains in the initial state.
Note that we therefore use a slightly modified model from now on: Since nodes start in arbitrary initial states, we do not consider the existence of an adversary as defined in the first parts of this paper, but simply let the adversary choose the initial state of each node.

\subsection{Protocol Description}
Intuitively, the goal of our protocol is to reach a round where all nodes are in state $\inactive$, because once this has been achieved the protocol behaves exactly as our fast protocol in the sense that the first node that switches from state $\inactive$ to state $\beep$ triggers a series of beeps throughout the whole graph, (equivalent to the scenario of the fast algorithm where the node that beeps first is activated by the adversary).
Since we already know that the fast algorithm converges in such a scenario, we are also guaranteed convergence in this setting as well.

We now describe the protocol in more detail.
In each round a node $v$ performs a consistency check (\Cref{algo:consistency_check}) before executing the main protocol (\Cref{algo:self_stabilizing_protocol:main}).

\begin{algorithm}[H]
\caption{Node $v$ checks \& resolves corrupted states in each round}
\label{algo:consistency_check}
\begin{algorithmic}[1]
	\State $b_1 \gets \s(v) = \beep \wedge (\delta(v) \in \mathit{CP} \vee (\delta(v)-1 \mod T) \in \mathit{CP})$
	\State $b_2 \gets \s(v) = \listen \wedge \delta(v) > 0$
	\If{$b_1 = false \wedge b_2 = false$} 
		\State $r(v) \gets 0$ \label{algo:reset_corrupted_state_1}
		\State $\s(v) \gets \pulse$ \label{algo:reset_corrupted_state_2}
	\EndIf
\end{algorithmic}
\end{algorithm}

\begin{algorithm}[H]
\caption{Pseudocode for the self-stabilizing protocol executed at each node $v$ in each round}
\label{algo:self_stabilizing_protocol:main}
\begin{algorithmic}[1]
	\If{$r(v) < \max\{4n,\mathit{sf}(q)+1\}$}
		\State $r(v) \gets r(v) + 1$ \Comment{Increment the round counter}
	\EndIf
	\If{$\s(v) = \inactive$}
	\If{$\exists w \in N(v): w$ beeps $\vee$ $r(v) \geq  4n$}
		\State $r(v) \gets 0$
		\State $b(v) \gets 1$
		\State Execute \Crefrange{algo:inactive_middle}{algo:inactive_end} of \Cref{algo:protocol} with $q \geq 5$
	\EndIf
	\ElsIf{$\s(v) = \beep$}
	\State Beep!
	\State $b(v) \gets b(v) + 1$
	\If{$b(v) \geq 4$}
		\State $r(v) \gets 0$
		\State $\s(v) \gets \pulse$
	\Else
		\State Execute \Crefrange{algo:delta_incr:0}{algo:delta_incr:1} of \Cref{algo:protocol} with $q \geq 5$
	\EndIf
	\ElsIf{$\s(v) = \listen$}
	\If{$\exists w \in N(v): w$ beeps}
		\State $b(v) \gets b(v) + 1$
		\If{$b(v) \geq 4 \vee r(v) > \mathit{sf}(q)$} \label{algo:self_stabilizing:error1}
			\State $r(v) \gets 0$
			\State $\s(v) \gets \pulse$
		\Else
			\State Execute \Crefrange{algo:check_checkpoint}{algo:no_induction} of \Cref{algo:protocol}  with $q \geq 5$
		\EndIf
	\Else
		\State $b(v) \gets 0$
		\State Execute \Crefrange{algo:increment_delta}{algo:reset_induced} of \Cref{algo:protocol}  with $q \geq 5$
	\EndIf
	\ElsIf{$\s(v) = \pulse$}
	\State Beep!
	\If{$r(v) \geq  4$}
		\State $r(v) \gets 0$
		\State $\s \gets \lock$
	\EndIf
	\ElsIf{$\s(v) = \lock$}
	\If{$r(v) \geq  4n$}
		\State $r(v) \gets 0$
		\State $\s \gets \inactive$
	\EndIf
	\EndIf
\end{algorithmic}
\end{algorithm}

In order to check its state for consistency, $v$ just checks if the value of $\delta(v)$ is valid when being in states $\beep$ or $\listen$.
More precisely, a node $v$ may be in state $\beep$ only if $\delta(v) = c$ or $\delta(v) = c+1$ where $c$ is a checkpoint.
Similarly, for state $\listen$ we allow any value $\delta(v) > 0$, since in legitimate states $v$ only beeps when $\delta(v)=0$.
If any of the above constrains is violated by the current assignments to $\delta(v)$ and $\s(v)$ then $v$ just sets its round counter $r(v)$ to $0$ and its state to $\pulse$ (\Cref{algo:reset_corrupted_state_1,algo:reset_corrupted_state_2}).
Note that we do not need consistency checks for the states $\pulse, \lock$ and $\inactive$, as we allow any arbitrary combination of variable assignments in these states.
The same holds for the variables $\induced(v)$, $b(v)$ and $r(v)$.

The main protocol (\Cref{algo:self_stabilizing_protocol:main}) is quite simple: In each round a node $v$ first increments its round counter $r(v)$ by $1$ (if it does not contain the maximum value yet).
Then $v$ performs operations based on $\s(v)$.
While in states $\beep$ or $\listen$, $v$ behaves exactly as in the fast algorithm from \Cref{sec:fast_sync} (using the set $\mathit{CP}_q$ instead of $\mathit{CP}$) as long as $v$ does not hear $4$ consecutive beeps and $v$ does not listen to a beep from a neighboring node when its round counter has reached a value at least $\mathit{sf}(q)+1$.
If one of the latter cases occurs then $v$ resets its round counter $r(v)$ to $0$ and switches to state $\pulse$.
A node $v$ that switched to state $\pulse$ beeps for $4$ consecutive rounds and switches to state $\lock$ afterwards (with $r(v)$ being reset to $0$ again).
Being in state $\lock$, $v$ just waits for $4n$ rounds and then switches to state $\inactive$ (with $r(v)$ being reset to $0$ again).
Nodes in state $\inactive$ are constantly listening for beeps for at most $4n$ rounds.
Upon receiving a beep or being in state $\inactive$ for $4n$ rounds the node $v$ resets $r(v), b(v)$ to $0$, $\delta(v)$ to $1$ and switches its state to $\beep$.

\section{Analysis of the Self-stabilizing Synchronization Protocol} \label{sec:self_stab_ana}
We prove \Cref{theorem:self_stabilizing} for the protocol described in \Cref{sec:self_stab_desc}.
For the analysis we assume that all nodes are by default in states modeled by the flags $b_1$ and $b_2$ in \Cref{algo:consistency_check}.
If a node is not in such a state, it can locally detect and resolve this via \Cref{algo:consistency_check}.

To simplify the analysis we introduce the set of \emph{super-states} $\mathit{ST} = \{\mathsf{PULSE}, \mathsf{LOCK}, \mathsf{INACTIVE},$ $\mathsf{FAST}\}$.
A node $v$ is in super-state $\mathsf{PULSE}$, $\mathsf{LOCK}$ or $\mathsf{INACTIVE}$ if and only if $\s(v) = \pulse, \lock$ or $\inactive$, respectively.
Also $v$ is in super-state $\mathsf{FAST}$ if and only if $\s(v) = \beep$ or $\listen$.
We denote the super-state of node $v$ in round $r$ by $\mathit{ST}_r(v)$.
Recall that a node being in super-state $\mathsf{FAST}$ basically executes the modified version of the fast algorithm (\Cref{algo:protocol}) with $q \geq 5$ unless there are abnormalities described in the previous section.

We can now formally define the set of legitimate states:

\begin{definition} \label{def:legitimate_state}
	At the beginning of some round $r$, the system is in a \emph{legitimate state} if the following three properties hold:
	\begin{itemize}
		\item[$(i)$] For all $v,w \in V$ it holds $\delta_r(v) = \delta_r(w)$.
		\item[$(ii)$] For all $v \in V$ it holds $\mathit{ST}_r(v) = \mathsf{FAST}$.
		\item[$(iii)$] For all $v \in V$ it holds $\induced_r(v) = \mathit{false}$.
	\end{itemize}
\end{definition}

While the first two properties of \Cref{def:legitimate_state} seem like natural requirements for the synchronization problem, the third property is needed because, intuitively, after all nodes have synchronized their clock values via the fast algorithm, in the round where the last node $v$ got its clock value in sync with the rest, $v$ still has $\induced(v)=\mathit{true}$ for $q$ more rounds.
This may lead to $v$ beeping at a checkpoint not equal to $0$, which may force any listening neighbor $w$ of $v$ with $r(w) = \mathit{sf}(q)+1$ to switch to state $\pulse$.
Thus when considering only the first two properties of \Cref{def:legitimate_state} the system would leave the legitimate state this way, violating the closure property of self-stabilization.

We first show that in case all nodes are in state $\mathsf{FAST}$ initially, then the system either converges, or we can detect an error.

\begin{lemma} \label{lemma:sfss:sf_pulse}
	Let the system be in a state where $\mathit{ST}(v) = \mathsf{FAST}$ for each node $v \in V$ holds.
	Then the system either converges to a legitimate state after $O(\max\{T,n\})$ rounds without nodes changing their super-states, or there exists at least one node $v$ that changes its super-state to $\mathsf{PULSE}$ after at most $O(\max\{T,n\})$ rounds.
\end{lemma}

\begin{proof}
	Assume to the contrary that the system does not converge to a legitimate state and nodes do not switch their super-states at all.
	Then after $\mathit{sf}(q) \in O(n)$ rounds, all nodes $v$ will have their round counter $r(v)$ set to $\mathit{sf}(q)+1$.
	As the system is assumed to not converge, there either exists at least one pair of neighboring nodes $(v,w)$ with different clock values or with different values of the $\induced_r$ flag.
	For the first case it holds that after at most $O(T)$ rounds, either $v$ or $w$ induces the other node.
	Let us assume that $v$ induces $w$.
	This triggers \Cref{algo:self_stabilizing:error1} of \Cref{algo:self_stabilizing_protocol:main} at $w$.
	Either way, $w$ switches to $\mathsf{PULSE}$ afterwards, which is a contradiction.
	Note that we need to wait these additional $O(T)$ rounds, since both $\induced$ flags at $v$ and $w$ may be set to $\mathit{false}$ initially.
	
	For the second case let $\induced_r(v) = \mathit{true}$ and $\induced_r(w) = \mathit{false}$.
	The it holds that after at most $q$ rounds that $v$ reaches the next checkpoint and beeps.
	In case $w$ beeps in that round as well the system converges which contradicts our initial assumption, so $w$ has to be in state $\listen$.
	Therefore, $w$ switches its super-state to $\mathsf{PULSE}$ as $r(w) = sf(q)+1$.
\end{proof}

For the above setting we can now show that all nodes get in super-state $\mathsf{LOCK}$:

\begin{lemma} \label{lemma:sfss:sf_lock}
	Let the system be in a state where $\mathit{ST}(v) = \mathsf{FAST}$ or $\mathit{ST}(v) = \mathsf{INACTIVE}$ for each node $v \in V$ holds.
	Assume that in round $r$ the node $v$ switches its super-state to $\mathsf{PULSE}$.
	Then after at most $4n$ rounds, all nodes are in super-state $\mathsf{LOCK}$.
\end{lemma}

\begin{proof}
	Once $v$ has switched to $\mathsf{PULSE}$, it beeps for the next $4$ rounds until it switches to $\mathsf{LOCK}$.
	As all neighbors $w$ of $v$ are in super-state $\mathsf{FAST}$ or $\mathsf{INACTIVE}$ it holds that $w$ will notice all $4$ beeps of $v$ in rounds where $w$ is listening to a beep.
	Due to the description of our algorithm $w$ beeps in rounds where it is not listening to a beep.
	Therefore $w$'s counter $b(w)$ increases by $1$ in each round, ultimately getting to $4$, which forces $w$ to switch to $\mathsf{PULSE}$ as well.
	Via an easy induction one can show that at most every $4$ rounds another node switches to $\mathsf{PULSE}$ as well, so after at most $4(n-1)$ rounds all nodes (except $v$) have switched to $\mathsf{PULSE}$.
	Thus after another $4$ rounds all nodes are in super-state $\mathsf{LOCK}$, which proves the lemma.
\end{proof}

We now show that we reach a state where all nodes are in super-state $\mathsf{LOCK}$.

\begin{lemma} \label{lemma:sfss:all_lock}
	Within the first $O(\max\{T,n\})$ rounds there is a round $r$ where it holds $\mathit{ST}(v) = \mathsf{LOCK}$ for all $v \in V$.
\end{lemma}

\begin{proof}
	Let $S_0 \subseteq V$ be the set of nodes that are in super-state $\mathsf{FAST}$ initially.
	For each connected component $C_0$ consisting of nodes in $S_0$ it either holds that after $O(\max\{T,n\})$ rounds all nodes in $C_0$ have synchronized, or there is a node that switches its super-state to $\mathsf{PULSE}$ (\Cref{lemma:sfss:sf_pulse}).
	In the latter case all nodes of $C_0$ are in super-state $\mathsf{LOCK}$ after $O(n)$ rounds (\Cref{lemma:sfss:sf_lock}).
	Consider only the connected components $C$ of $S_0$ that got in sync.
	Nodes that switch their super-state from $\mathsf{LOCK}$ to $\mathsf{INACTIVE}$ and that are neighbors of nodes within $C$ are either being included into $C$ via an induced beep from one of their neighbors in $C$, or they force an error in case they induce a node $v \in C$ whose round counter $r(v)$ is already at $\mathit{sf}(q)+1$.
	If only the first case holds, then all $n$ nodes get included into one single connected component that is in sync after $O(n)$ rounds.
	In the second case the nodes in $C$ switch to $\mathsf{PULSE}$ and, consequently, to $\mathsf{LOCK}$ afterwards.
	At most $4|C|$ rounds are needed for this to happen due to nodes being in super-state $\mathsf{PULSE}$ for exactly $4$ rounds.
	Using the same arguments as above it also holds for two neighboring nodes $v,w \in C$ that once they switched to $\mathsf{LOCK}$, the values of their round counters differs by at most $4$.
	We arrive at a state where no node from the set $S_0$ is in $\mathsf{FAST}$ anymore.
	Now nodes $v \not \in S_0$ get into super-state $\mathsf{FAST}$ and try to synchronize their clocks.
	However, as they start their round counter $r(v)$ at $0$ upon entering $\mathsf{FAST}$, new errors are only detected after at least $\max\{T,\mathit{sf}(q)\}$ rounds.
	By that time all nodes $v \not \in  S_0$ are either in super-state $\mathsf{INACTIVE}$ or in $\mathsf{FAST}$.
	As $\mathit{sf}(q) > 4n$ it holds that nodes in $S_0$ switched from super-state $\mathsf{LOCK}$ to $\mathsf{INACTIVE}$.
	This means that all nodes are either in super-state $\mathsf{INACTIVE}$ or in $\mathsf{FAST}$.
	Now, once inactive nodes get into super-state $\mathsf{FAST}$ either the system fully gets in sync or an error will be detected due to round counters being too large at some nodes.
	This implies that \textit{all} $n$ nodes will switch to $\mathsf{PULSE}$ and then to $\mathsf{LOCK}$ after at most $4n$ rounds (\Cref{lemma:sfss:sf_lock}).
	Thus we arrive at a state where all nodes are in super-state $\mathsf{LOCK}$.
\end{proof}

We can finally prove \Cref{theorem:self_stabilizing}:

\begin{proof}[Proof of \Cref{theorem:self_stabilizing}]
	We know by \Cref{lemma:sfss:all_lock} that we reach a state where it holds $\mathit{ST}_r(v) = \mathsf{LOCK}$ for all nodes $v \in V$ after $O(\max\{T,n\})$ rounds.
	Once we reached such a state it holds $\mathit{ST}_r(v) = \mathsf{INACTIVE}$ after at most $4n$ additional rounds.
	Therefore we can just simply apply the analysis for the fast algorithm to our protocol at this point, with the modification that a single period now consists of $q \geq 5$ rounds instead of $4$.
	Due to \Cref{cor:super_fast_ext} we need at most $O(D)$ additional rounds until all clock values are in sync.
	As all nodes $v$ start with $r(v) = 0$ when $v$ leaves the $\mathsf{INACTIVE}$ super-state, $v$ will keep executing the fast algorithm for at least $\mathit{sf}(q) = O(n)$ rounds, which is enough to get all nodes in sync with regards to their $\delta$-values, $\s$-values and $\induced$-values.
	Therefore the convergence property is satisfied.
	 
	For closure it is easy to see that all nodes only beep (maturely) in the same round every $T$ rounds (i.e, once their clock value is at $0$), thus keeping not only their $\delta$-values but also their $\s$ and $\induced$ values in sync.
\end{proof}

\section{Asynchronous Slot Boundaries} \label{sec:async_model}
We want to briefly argue why our algorithms also work in a model where time is continuous and divided into slots of $\mu$ time units (we assume that $\mu$ is known to all nodes).
Initially, the slot boundaries may not be synchronized among the nodes.
Nodes that beep do this for the whole duration of their time slot.
We also assume that beeps are instantaneous, i.e., they are immediately received by listening neighboring nodes.

\begin{figure}[ht]
	\centering
	\includegraphics[scale=1.2]{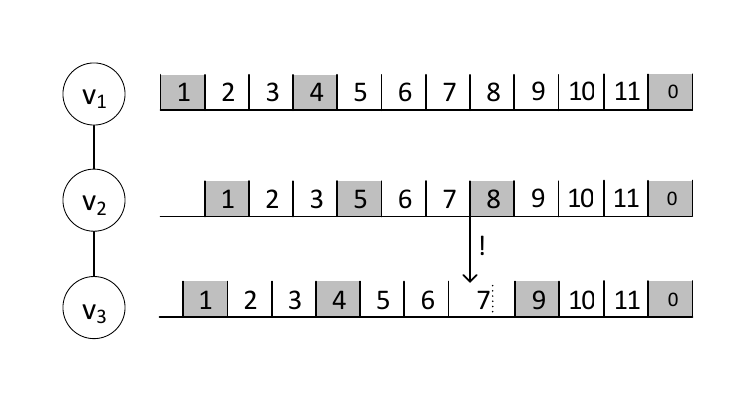}
	\caption{Example for a simple line with $3$ nodes ($T=12$). Slots highlighted in grey represent a beep, the number of a slot represents the $\delta$-value of the corresponding node. Node $v_1$ got activated by the adversary, $v_2$ got activated by $v_1$ and $v_3$ got activated by the adversary in between the activation of $v_1$ and $v_2$. When $\delta(v_3) = 7$ it gets induced by $v_2$ and extends its slot boundary to fit to the start of $v_2$'s beep, which causes all slot boundaries to be in sync.}
	\label{fig:slots}
\end{figure}

We apply the following rule to our protocols: Let $t \in [0,\mu)$ be the exact point in time where $v$ starts recognizing a beep for the remaining time of its slot.
In case $\delta_t(v)=c-1$ for some checkpoint $c \in \mathit{CP}$ or $v$ is $\mathsf{INACTIVE}$, $v$ extends its current slot by $t$ time units.
By doing so, $v$ gets its slot boundaries for the next slot in sync with the slot boundaries of the node(s) by which it got induced at time $t$ in its current slot.
Ultimately, all nodes not only synchronize their clock values to the value of the node that got activated first, but also their slot boundaries (see \Cref{fig:slots} for an example).

\section{Conclusion}
We presented new algorithms for clock synchronization in the beeping model.
For future work one may investigate when or under which circumstances having nodes be able to send more than one bit per message actually helps improving the runtimes of our algorithms.

Also, since our self-stabilizing protocol assumes knowledge of some $N \in \Theta(n)$ by the nodes, an important question would be if there is a self-stabilizing solution that allows nodes start with arbitrary estimates for $n$ (and possibly also for $T$).
Solving this problem allows extensions to extensions to dynamic networks, where nodes may join and leave.

On the same note, a node $v$ that joins the system may trigger the other nodes to synchronize their clock values to the one of $v$ (depending on $v$'s initial clock value).
It would be interesting to see if our algorithms can be modified in order to support joining of nodes more effective, i.e., in a constant amount of rounds, while still preserving the overall runtime bounds for the initial synchronization.

\bibliographystyle{alpha}
\bibliography{literature}

\end{document}